\documentclass[11pt,a4paper,reqno,twoside]{article}

\usepackage{amssymb,comment}
\usepackage{amsmath}
\usepackage{float}
\usepackage{latexsym}
\usepackage{amsthm}
\usepackage{empheq}
\usepackage{bm}
\usepackage{booktabs}
\usepackage{subcaption}
\usepackage{tikz,pgfplots,cite}
\usepackage{balance}

\usepackage{cite}
\usepackage{amsmath,amssymb,amsfonts}
\usepackage{algorithmic}
\usepackage{graphicx}

\usepackage{latexsym}
\usepackage{thm-restate}
\usepackage{empheq}
\usepackage{bm}
\usepackage{array}
\usepackage{hyperref}
\usepackage[margin=2.8cm]{geometry}
\usepackage{authblk}

\theoremstyle{definition}
\newtheorem{theorem}{Theorem}[section]
\newtheorem{corollary}[theorem]{Corollary}
\newtheorem{lemma}[theorem]{Lemma}
\newtheorem{remark}[theorem]{Remark}
\newtheorem{definition}[theorem]{Definition}

\newtheorem{proposition}[theorem]{Proposition}
\newtheorem{notation}[theorem]{Notation}
\newtheorem{example}[theorem]{Example}
\newtheorem{problem}{Problem}
\newtheorem{conjecture}[theorem]{Conjecture}

\newcommand{\F}{\mathbb{F}}
\newcommand{\N}{\mathbb{N}}

\newcommand{\E}{\mathbb{E}}
\newcommand{\mC}{\mathcal{C}}

\newcommand{\qbinom}[2]{\genfrac{[}{]}{0pt}{}{#1}{#2}}
\title{The DNA Coverage Depth Problem: Duality, \\ Weight Distributions, and Applications}

\author[1]{Matteo Bertuzzo}
\author[1]{Alberto Ravagnani}
\author[2]{Eitan Yaakobi}

\affil[1]{Eindhoven University of Technology, Eindhoven, The Netherlands}
\affil[2]{Technion -- Israel Institute of Technology, Haifa, Israel}

\begin{document}
\maketitle

\renewcommand{\thefootnote}{\fnsymbol{footnote}}
\footnotetext{The research of M.B. is partially supported by the EuroTech Program. The research of A.R. is supported by the Dutch Research Council through grants VI.Vidi.203.045, OCENW.KLEIN.539, and by the European Commission. The research of E.Y. is funded by the European Union (ERC, DNAStorage, 101045114). 

Views and opinions expressed are those of the authors only and do not necessarily reflect those of the European Union or the European Research Council Executive Agency. Neither the European Union nor the granting authority can be held responsible for them. 

Some of the results in this paper have been presented at ISIT 2025.

\textit{Email addresses:}
\href{mailto:m.bertuzzo@tue.nl}{m.bertuzzo@tue.nl},
\href{mailto:a.ravagnani@tue.nl}{a.ravagnani@tue.nl},
\href{mailto:yaakobi@cs.technion.ac.il}{yaakobi@cs.technion.ac.il}.
}
\renewcommand{\thefootnote}{\arabic{footnote}}

\begin{abstract}
The coverage depth problem in DNA data storage is about computing the expected number of reads needed to recover all encoded strands. Given a  generator matrix of a linear code, this quantity equals the expected number of randomly drawn columns required to obtain full rank. While MDS codes are optimal when they exist, i.e., over large fields, practical scenarios may rely on structured code families defined over small fields.
In this work, we develop combinatorial tools to 
solve the DNA coverage depth problem for various linear codes, based on duality arguments and the notion of extended weight enumerator. Using these methods, we derive closed formulas for the simplex, Hamming, ternary Golay, extended ternary Golay, and first‑order Reed–Muller codes. The centerpiece of this paper is a general expression for the coverage depth of a linear code in terms of the weight distributions of its higher-field extensions.  
\end{abstract}

\medskip

\section{Introduction}

The volume of digital data produced worldwide continues to grow at a rapid pace, creating an ever-increasing demand for storage that already exceeds the current supply, with the gap continuing to widen~\cite{b1}. This trend highlights the urgent need for innovative storage technologies that offer higher density, improved efficiency, and long-term durability beyond the limits of existing solutions. In this context, DNA-based storage systems provide an attractive solution, particularly for long-term data archiving, due to their durability, compactness, and low maintenance costs~\cite{b2},~\cite{b3}. 

Data storage in DNA involves a multi-step pipeline. First, the original data is encoded from a string of bits into sequences based on the DNA alphabet ${A, C, G, T}$. These sequences are subsequently divided into blocks, which are synthesized into artificial DNA molecules, referred to as \textit{strands}.
Due to current technology limitations, synthesized strands are limited to lengths of up to approximately $300$ bases~\cite{synth_length}, and multiple noisy copies of each encoded strand are generated. The resulting DNA molecules are then stored within a container. To retrieve the stored information, the strands are translated back into DNA sequences through the sequencing process, generating multiple unordered copies, called \textit{reads}. 
In the last step, the reads are decoded to reconstruct the user’s original data. A distinctive feature of this process is that sequencing accesses strands \textit{randomly}.

Despite its significant potential~\cite{b4, b5, b6, b7, b8, b9, b10}, DNA-based storage currently faces practical limitations due to its relatively slow throughput and high costs compared to alternative storage techniques, resulting from the efficiency of DNA sequencers~\cite{b2, b11, b12}. These limitations are related to the notion of \textit{coverage depth}~\cite{b13}, defined as the ratio between the number of sequenced reads and the number of designed DNA strands. This quantity plays a key role in determining sequencing costs and system performance. 

In this scenario, a central question is how many reads are needed on average to recover all encoded strands. This is known as the \emph{DNA coverage depth problem} and its coding theory aspects have been recently introduced in~\cite{cover_your_bases}. When the data is encoded by a linear block code, the coverage depth problem can be formulated in algebraic terms: Given a rank $k$ generator matrix $G$ of a linear code $\mC$, the number of reads to recover the encoded strands equals the expected number of randomly drawn columns required to obtain a rank $k$ matrix. 
Recent work has shown that 
if the information strands are encoded using an MDS code, the expectation is $n(H_n - H_{n-k})$, where $H_i$ is the $i$-th harmonic number. In particular, MDS codes achieve the optimal value. However, these codes only exist over large finite fields. 

Recent years have seen growing interest in the coverage depth problem and related models. Beyond the full‑recovery setting of~\cite{cover_your_bases}, several works have addressed the \emph{random access} scenario, in which the goal is to retrieve a single information strand efficiently~\cite{rel1, rel2, a_combinatorial_perspective, rel3, rel4, rel5}. Other extensions include models where groups of strands collectively represent one file~\cite{other_rel_1} and studies of read/write trade‑offs for practical coding architectures~\cite{other_rel_2}. Further generalizations have been developed for combinatorial shortmer‑based systems~\cite{ext1, ext2,ext3} and for storage schemes using composite DNA letters~\cite{ext4}.

Motivated by the results of~\cite{cover_your_bases}, this work investigates the coverage depth problem for linear codes defined over small finite fields. Our main objective is to understand how the structural properties of a code
govern the expected number of reads required for full data recovery,
a research line that was initiated in~\cite{isit_paper}. To this end, we develop a set of combinatorial and algebraic tools centered around information‑set enumeration, duality identities, and the extended weight enumerator.

These tools allow us to derive closed formulas for a variety of classical code families. We show how simplex codes admit a simple expression for the expectation. Using a duality argument, we then obtain explicit formulas for Hamming codes and for both the ternary Golay and extended ternary Golay codes. Our main theoretical contribution is a general expression of the coverage depth in terms of the weight distributions of the first few higher-field extensions of the underlying code. This result reduces the computation of the expectation to weight enumeration. As an application, we use the known extended weight enumerator of first‑order Reed–Muller codes to compute their coverage depth as a closed formula.

The remainder of the paper is organized as follows.
Section~\ref{problem_statement} introduces the coverage depth problem and its algebraic formulation. 
Section~\ref{simplex_section} computes the expectation for simplex codes. 
Section~\ref{duality_section} establishes a duality identity expressing the expectation 
in terms of the information-sets of the dual code. 
This result is then used in Section~\ref{Hamming_section} to derive closed formulas for Hamming and ternary Golay codes. 
Section~\ref{weight_distr_section} presents our main result: A general expression of the coverage depth in terms of the weight distributions of the higher-field extension codes of 
$\mC$. 
Finally, Section~\ref{reedmuller_section} applies this formula to obtain an explicit expression for first‑order Reed–Muller codes.

\section{Problem Statement} \label{problem_statement}
In this paper, $q$ is a prime power and $\mathbb{F}_q$ is the finite field with $q$ elements. We let $k$ and $n$ be positive integers with $2 \leq k \leq n$. Furthermore, for a positive integer $m$, we
denote by $H_m$ the $m$-th harmonic number: 
$$H_m=\sum_{i=1}^m \frac{1}{i}.$$

In a typical DNA-based storage system, data is stored as a length-$k$ vector whose entries are themselves vectors (called \textit{strands}) of length $\ell$ over the alphabet $\Sigma = \{ A,C,G,T \}$.
To allow using coding theory tools, we embed $\Sigma^\ell$ into $\mathbb{F}_q$ and use a $k$-dimensional linear block code $\mathcal{C} \subseteq \mathbb{F}_q^n$ to encode an information vector $(x_1, \dots, x_k) \in \mathbb{F}_q^k$ to an encoded vector $(y_1, \dots, y_n) \in \mathbb{F}_q^n$. Note that we would need $|\Sigma^\ell| = 4^\ell$ to divide $q$ in order to identify $\Sigma^\ell$ with~$\F_q$; however, in this paper we do not consider any restrictions on $q$, allowing it to be any prime power.

If a user wishes to retrieve the stored information, the encoded strands initially undergo an amplification process, followed by sequencing.
This generates various unsorted copies of each strand, which may contain errors compared to the originals. These are called \textit{reads}. To simplify our analysis, in this paper we will assume that no errors are made in any of these steps, hence the final output of the process is a multiset of reads, obtained without any specified order. 

The starting point of this paper is a result about the coverage depth problem for DNA data storage~\cite{cover_your_bases}, when all information strands need to be recovered. Since the $k$ information strands are encoded using a generator matrix~$G \in \mathbb{F}_q^{k \times n}$, there is a one-to-one correspondence between the encoded strands and the columns of~$G$, namely the $i$-th encoded strand corresponds to the $i$-th column of the generator matrix; therefore, recovering the $i$-th information strand is equivalent to recovering the $i$-th standard basis vector, that is, it must be in the span of the already recovered columns of $G$, since we can see these columns as vectors in $\mathbb{F}_q^k$. Motivated by these results, we  define the first problem studied in this paper.

\begin{problem}[\textbf{The DNA coverage depth problem}] \label{first_problem}
    Let $G \in \mathbb{F}_q^{k \times n}$ have rank~$k$. Suppose that the columns of $G$ are drawn uniformly randomly with repetition, meaning that each column can be drawn multiple times. Compute the expected number of columns one needs to draw until all the standard basis vectors are in their $\mathbb{F}_q$-span (or equivalently until the submatrix formed by the drawn columns has rank~$k$). We denote such expectation by $\mathbb{E}[G]$.
\end{problem}

In general, obtaining a closed expression for $\E[G]$ is challenging. 
The difficulty arises from the intricate dependence among successive draws: 
Unlike in the classical coupon collector problem~\cite{ccp1, ccp2, ccp3, ccp4}, here the effect of a newly drawn column depends on the \textit{span} of all the previously drawn ones. In particular, 
drawing a coupon that was not drawn previously is \textit{not} necessarily making progress
towards completing the collection, since the new column may fail to increase the rank.

We start by showing that, in contrast to the random access coverage depth problem~\cite[Problem~1]{a_combinatorial_perspective}, where only a single information strand is to be recovered, the value of $\E[G]$ only depends on the row-space of $G$, i.e., on the code that the matrix~$G$ generates.

\begin{proposition} \label{independence_from_G_proposition}
    Let $G,G' \in \F_q^{k \times n}$ have the same row-space. Then $\E[G]=\E[G']$.
\end{proposition}
\begin{proof}
Since $G$ and $G'$ have the same row-space, there exists an invertible matrix $A \in \F_q^{k \times k}$ with $G'=AG$. 
The statement follows from the fact that multiplying by $A$ preserves the linear dependencies among columns.
\end{proof}

Proposition~\ref{independence_from_G_proposition} 
shows that the quantity 
$\E[\mC]$ is well defined for a linear error-correcting code $\mC \subseteq \F_q^n$ as $\E[G]$, where $G$ is \textit{any} generator matrix of $\mC$. Therefore, we will use the symbols $\E[G]$ and $\E[\mC]$ interchangeably. We are now ready 
to state the second problem which we will address in this paper, strictly related to Problem~\ref{first_problem}.

\begin{problem}[\textbf{The optimal coverage depth problem}] \label{second_problem}
    For given values of $n$, $k$ and $q$, compute the value
$$\mathbb{E}_{\textnormal{opt}}[n,k]_q \triangleq \min\{\mathbb{E}[\mC] \, : \,  \mC \subseteq \F_q^n \, \text{ is a code of dimension $k$}\},$$
    and construct
    a code $\mC$ attaining the minimum. 
\end{problem}

Throughout the paper, $\smash{G \in \mathbb{F}_q^{k \times n}}$ denotes a rank $k$ matrix, $\mC \subseteq \mathbb{F}_q^n$ is the (linear, block) code of dimension $k$ generated by $G$, and $\mC^\perp \subseteq\F_q^n$ is the dual code of $\mC$ of dimension $n-k$.

Various results related to Problems~\ref{first_problem} and~\ref{second_problem} were obtained in~\cite{cover_your_bases}. We mention the most relevant one for this paper.

\begin{theorem}[\text{see \cite[Corollary~1]{cover_your_bases}}] \label{thm_MDS_bound}
    For any generator matrix~$G$ of a code $\mC \subseteq \F_q^n$ of dimension $k$ we have 
    \[
    \E[G] \ge \sum_{i=0}^{k-1} \frac{n}{n-i} = n(H_n - H_{n-k}).
    \]
    Furthermore, the lower bound is attained with equality only by any generator matrix of an MDS code. 
\end{theorem}

\begin{remark}
    Theorem~\ref{thm_MDS_bound} provides a lower bound on the expectation and solves Problem~\ref{first_problem} for MDS codes. It also solves Problem~\ref{second_problem} for any choice of parameters $n, k$ and $q$ such that there exists an MDS code $\mC \subseteq \F_q^n$ of dimension $k$. In particular, assuming that the MDS conjecture~\cite{mds1} holds, we can write
    \[
    \mathbb{E}_{\text{opt}}[n,k]_q = n(H_n - H_{n-k}) \text{ when } q \ge n-1.
    \]
\end{remark}

It is well known that MDS codes only exist over sufficiently large finite fields. In fact, it has been conjectured (and proven in several instances) that $q \ge n-1$ is a necessary condition for the existence of an MDS code, with the exception of very few parameter sets that require ``only'' $q \ge n-2$; see~\cite{mds1,mds2,mds3,mds4}.
It is therefore natural to investigate what results can be achieved in this context when $q$ is too small to allow the existence of an MDS code.

\section{Performance of Simplex Codes} \label{simplex_section}
When focusing on small finite fields, it is natural to consider simplex codes. Recall that a generator matrix $G$ of the $k$-dimensional simplex code over $\F_q$ is obtained by organizing all nonzero elements of $\F_q^k$, up to nonzero scalar multiples, as columns of $G$. The simple structure of this generator matrix makes it possible to obtain a closed formula for $\E[G]$ using the $q$-analogue of a standard argument for the coupon collector's problem.

\begin{theorem} \label{thm_E_k_simplex}
    Let $\mathcal{C} \subseteq \mathbb{F}_q^n$ be the $q$-ary simplex code of dimension $k$, where $n = (q^k-1)/(q-1)$. We have 
    \begin{equation*} 
    \mathbb{E}[\mathcal{C}] = k + \sum_{i=1}^k \frac{q^{i-1}-1}{q^k - q^{i-1}}.
    \end{equation*}
\end{theorem}
\begin{proof}
    Fix any generator matrix $G$ of $\mC$. For $i \in \{1, \ldots, k\}$, let $s_i(\mC)$ be the random variable that governs the number of draws until the selected columns span a space of dimension $i$, when the columns previously drawn span a space of dimension $i-1$. Note that the expected value of $s_1(\mC)$ is equal to 1, since all columns of $G$ are nonzero.
    Since the columns of $G$ are the elements of $\F_q^k$ up to multiples, 
    $s_i(\mathcal{C})$ is a geometric random variable with success probability
    \[
    \frac{n - \frac{q^{i-1}-1}{q-1}}{n}.
    \]
    By the linearity of expectation, we therefore have
    \[
    \begin{split}
    \mathbb{E}[\mathcal{C}] &= \mathbb{E}\left[ \sum_{i=1}^k s_i(\mathcal{C}) \right] = \sum_{i=1}^k \mathbb{E}[s_i(\mathcal{C})] 
    = k + \sum_{i=1}^k \frac{q^{i-1}-1}{q^k - q^{i-1}},
    \end{split}
    \]
    as desired.
\end{proof}

Based on experimental evidence, the simplex code performs best among all codes with the same parameters. However, we still do not have a formal proof of this fact.

\begin{conjecture}
    Let $\mC \subseteq \mathbb{F}_q^n$ be the $q$-ary simplex code of dimension $k$, where $n = (q^k-1)/(q-1)$. Then $\mC$ solves Problem~\ref{second_problem}.
\end{conjecture}

\section{A Duality Result} \label{duality_section}

A natural question is how the value of $\E[\mC]$ relates to the value $\E[\mC^\perp]$. While these two quantities do not determine each other in general, in this section we derive a duality result
that expresses $\E[\mC]$ in terms of the combinatorial structure of the dual code $\mC^\perp$. In Section~\ref{Hamming_section} we will illustrate how this result 
has a concrete application, namely computing
$\E[\mC]$ when~$\mC$ is the Hamming code.

\begin{definition}
    Let $\mathcal{C} \subseteq \F_q^n$ be a $k$-dimensional code and let $\pi_S: \F_q^n \to \F_q^{|S|}$ denote the projection onto the coordinates indexed by a set $S \subseteq\{1, \ldots, n\}$. We say that 
    $S$ is an \textbf{information set} for $\mC$ if $\pi_S(\mC)$ has dimension $k$.

    If $G$ is a generator matrix of $\mC$,
    then for $0 \leq s \leq n$ we denote by $g_j$ the $j$-th column of~$G$, and define
    \[
    \alpha(\mC,s) = \left|\left\{ S \subseteq \{1, \ldots, n\} \, : \, |S| = s, \, \langle g_j \, : \, j \in S \rangle = \F_q^k\right\}\right|,
    \]
    which does not depend on the choice of $G$ and counts the number of information sets of cardinality $s$ of $\mC$. 
\end{definition}

Note that $\alpha(\mC,s) = 0$ for $0 \leq s \leq k-1$.
Following the same reasoning as Proposition~\ref{independence_from_G_proposition},
it can be checked that 
$\alpha(\mC,s)$ only depends
on the code $\mathcal{C}$ that $G$ generates. We will therefore use the symbols $\alpha(\mathcal{C},s)$ and $\alpha(\mC,s)$ interchangeably.

We start by expressing $\mathbb{E}[\mC]$ in terms of the values $\alpha(\mathcal{C},s)$ we just introduced.
The proof is similar to that of~\cite[Lemma~1]{a_combinatorial_perspective} and is therefore omitted in this paper.

\begin{proposition} \label{E_k[C]_general_formula}
For any $k$-dimensional code $\mC \subseteq \F_q^n$
we have 
    \[
    \mathbb{E}[\mathcal{C}] = nH_n -  \sum_{s=k}^{n-1} \frac{\alpha(\mathcal{C},s)}{\binom{n-1}{s}}.
    \]
\end{proposition}

Before continuing, we illustrate how Proposition~\ref{E_k[C]_general_formula} can be used to easily compute the expectation for MDS codes.

\begin{example}
    Let $\mC \subseteq \F_q^n$ be an MDS code of dimension $k$ and let $G$ be a generator matrix of $\mC$. Since $G$ is an MDS matrix, every $k$ columns of $G$ are linearly independent. Thus, we have that
    \[
    \alpha(\mC,s) = \binom{n}{s} \quad \textnormal{for $k \le s \le n$}.
    \]
    By substituting these values into the formula of Proposition~\ref{E_k[C]_general_formula}, we obtain
    \[
    \mathbb{E}[\mC] = nH_n - \sum_{s=k}^{n-1} \frac{\binom{n}{s}}{\binom{n-1}{s}},
    \]
    which simplifies to $n(H_n - H_{n-k})$ after straightforward computations.
\end{example}

We now turn to the duality result.
We relate 
the value of $\alpha(\mathcal{C},s)$ to the structure of the dual code $\mC^\perp$. To do so,
it is convenient to introduce some auxiliary quantities. We denote the Hamming support of a vector $x \in \F_q^n$ as $\sigma(x)=\{i \, :\, x_i \neq 0\}$.
For a code $\mC \subseteq \F_q^n$ and a subset $S \subseteq \{1, \ldots, n\}$, we let $\mC(S)=\{x \in \mC \, : \, \sigma(x) \subseteq S\}$. The complement of a set $S$ is denoted by $S^{\mathsf{c}}=\{1, \ldots, n\} \setminus S$.

\begin{notation} \label{beta_definition}
For a code $\mC \subseteq \F_q^n$, $1 \le \ell \le k$, and $0 \le s \le n$, let
\begin{equation*} 
\beta_{\ell}(\mathcal{C},s) = |\{ S \subseteq \{1, \ldots, n\} \, : \, |S| = s, \, \dim(\mathcal{C}(S^{\mathsf{c}}))=\ell \}|.
\end{equation*}
\end{notation}

The main tool of this section is the following result.

\begin{lemma} \label{duality_lemma}
    Let $\mathcal{C} \subseteq \F_q^n$ be
    a code of dimension $k$. For all $\ell, s$ with
    $1 \le \ell \le k$ and $0 \le s \le n$ we have
    \begin{equation} \label{beta_relation}
    \beta_{\ell}(\mathcal{C},s) = \beta_{\ell+s-k}(\mathcal{C}^\perp,n-s).
    \end{equation}
    In particular,
    \[ 
    \alpha(\mathcal{C}, s) = \beta_{s-k}(\mathcal{C}^{\perp},n-s).
    \]
\end{lemma}
\begin{proof}
    Take an arbitrary set $S \subseteq \{1, \ldots, n\}$ of cardinality $s$. Using the rank-nullity theorem we obtain
    \begin{equation} \label{rank_nullity}
    \dim(\pi_S(\mathcal{C})) + \dim(\ker(\pi_S)) = k,
    \end{equation}
    which we can rewrite as
    \begin{equation} \label{rank_nullity_2}
        \dim(\pi_S(\mathcal{C})) + \dim(\mathcal{C}(S^{\mathsf{c}})) = k.
    \end{equation} 
    Moreover, by~\cite[Theorem~24]{ravagnani} we have 
    \[
    |\mathcal{C}(S)| = \frac{|\mathcal{C}|}{q^{n-s}} |\mathcal{C}^{\perp}(S^{\mathsf{c}})|,
    \]
    i.e.,
    \begin{equation} \label{rank_nullity_dual}
    \dim(\mathcal{C}(S)) = k-n+s+\dim(\mathcal{C}^{\perp}(S^{\mathsf{c}})).
    \end{equation}
    Therefore, from~\eqref{rank_nullity} we know that $\dim(\pi_S(\mathcal{C})) = t$ if and only if $\dim(\mathcal{C}(S^{\mathsf{c}})) = k-t$, and~\eqref{rank_nullity_dual} tells us that
    the latter equality is equivalent to $\dim(\mathcal{C}^{\perp}(S)) = s-t$. All of this shows that there exists a bijection
    \begin{multline*}
    \{ S \subseteq  \{1, \ldots, n\} \, : \,  |S| = s, \, \dim(\mathcal{C}(S^{\mathsf{c}})) = k-t \} \\  \quad \to \{ S \subseteq \{1, \ldots, n\} \, : \, |S| = n-s, \, \dim(\mathcal{C}^{\perp}(S^{\mathsf{c}})) = s-t  \},
    \end{multline*}
    from which we obtain the first part of the lemma. For the second part, it suffices to use the fact that
    \begin{equation} \label{dual_eq_alpha_beta0}
        \alpha(\mathcal{C},s) = \beta_0(\mathcal{C},s),
    \end{equation}
    which easily follows from the definitions.
    Combining this equality with \eqref{beta_relation} we obtain the second part of the lemma.
\end{proof} 

The next corollary follows directly from Proposition \ref{E_k[C]_general_formula} and Lemma \ref{duality_lemma}. We will use it in the next section.

\begin{corollary}
    Let $\mathcal{C} \subseteq \F_q^n$ be
    a code of dimension $k$. We have 
    \begin{equation} \label{exp_in_terms_of_beta}
        \mathbb{E}[\mathcal{C}] = nH_n -  \sum_{s=k}^{n-1} \frac{\beta_{s-k}(\mathcal{C}^{\perp},n-s)}{\binom{n-1}{s}}.
    \end{equation}
\end{corollary}

\section{Performance of Hamming and Ternary Golay Codes} \label{Hamming_section}

We apply the results of Section~\ref{duality_section} to compute the value of $\E[\mC]$, where $\mC$ is the Hamming code, 
the ternary Golay code, and the extended ternary Golay code.

The following result represents a direct application of Lemma \ref{duality_lemma} and allows us to obtain the value of the expectation of the Hamming code in terms of its corresponding dual code, namely, the simplex code.

\begin{theorem} \label{thm_E_k_Hamming}
    Let $\mathcal{C} \subseteq \mathbb{F}_q^n$ be the $q$-ary Hamming code of redundancy $r$, where $n=(q^r-1)/(q-1)$. We have 
    \[
    \mathbb{E}[\mathcal{C}] = nH_n -  \sum_{\ell=1}^{r} \frac{1}{\binom{n-1}{n-\ell}}\frac{\prod_{i = 0}^{\ell-1} \frac{q^r-q^i}{q-1}}{\ell !}.
    \]
\end{theorem}
\begin{proof}
    We have that $\dim(\mC) = n-r$, hence we can rewrite \eqref{exp_in_terms_of_beta} as 
    \[
    \mathbb{E}[\mathcal{C}] = n H_n - \sum_{\ell=1}^{r} \frac{\beta_{r-\ell}(\mathcal{C}^{\perp},\ell)}{\binom{n-1}{n-\ell}},
    \]
    where the dual code $\mathcal{C}^{\perp}$ is the $[n,r]_q$ simplex code. Applying~\eqref{rank_nullity_2} to Notation~\ref{beta_definition} we obtain
    \[
    \beta_{r-\ell}(\mathcal{C}^\perp,\ell) =  |\{ S \subseteq \{1, \ldots, n\} : |S| = \ell, \dim(\pi_S(\mathcal{C}^\perp))=\ell \}|.
    \]
    It remains to count the number of subsets of cardinality $\ell$ whose corresponding columns are linearly independent. To do this, we use again the fact that the columns of any generator matrix of the simplex code are all the nonzero vectors of $\mathbb{F}_q^r$ up to nonzero scalar multiples. Hence we have
    \[
    \beta_{r-\ell}(\mathcal{C}^\perp,\ell) = \frac{\prod_{i=0}^{\ell-1} \Big(\frac{q^r-1}{q-1} - \frac{q^i-1}{q-1}\Big)}{\ell!} = \frac{\prod_{i=0}^{\ell-1} \frac{q^r-q^i}{q-1}}{\ell!},
    \]
    from which the statement follows.
\end{proof}

We now turn to ternary Golay codes.
We start by establishing a refinement of the formula in Proposition \ref{E_k[C]_general_formula}, which involves the minimum distance of the code and results in fewer unknown values to compute. 

\begin{corollary} \label{E[C]_general_enhanced}
    Let $\mC \subseteq \F_q^n$ be a code of dimension $k$ and minimum distance $d$. We have
    \[
    \E[\mC] = n(H_n - H_{d-1}) -  \sum_{s=k}^{n-d} \frac{\alpha(\mathcal{C},s)}{\binom{n-1}{s}}.
    \]
\end{corollary}
\begin{proof}
    From \eqref{dual_eq_alpha_beta0} it follows that
   $    \alpha(\mathcal{C},s) = | \{ S \subseteq [n] \, : \,  |S| = s, \, \dim(\mathcal{C}(S^{\mathrm{c}})) = 0 \} |$.
    This means that, if $n-s \leq d-1$,
    we have $\dim(\mathcal{C}(S^{\mathrm{c}})) = 0$,  and thus $\alpha(\mathcal{C},s) = \binom{n}{s}$. 
    Therefore,
    \[
    \begin{split}
    \E[\mC] & = nH_n -  \sum_{s=k}^{n-1} \frac{\alpha(\mathcal{C},s)}{\binom{n-1}{s}} \\
    & = nH_n -  \Bigg(\sum_{s=k}^{n-d} \frac{\alpha(\mathcal{C},s)}{\binom{n-1}{s}} + \sum_{s=n-d+1}^{n-1} \frac{n}{n-s}\Bigg) \\
    & = n(H_n - H_{d-1}) -  \sum_{s=k}^{n-d} \frac{\alpha(\mathcal{C},s)}{\binom{n-1}{s}}. \qedhere
    \end{split}
    \]
\end{proof}

When the code $\mC$ is the ternary Golay code or the extended ternary Golay code, Corollary~\ref{E[C]_general_enhanced} provides a formula that requires the computation of only one of the $\alpha(\mC,s)$ values, thereby making Problem~\ref{first_problem} easier to solve for these two codes. We will utilize the weight enumerator of a code, which appears in this section and features even more prominently in Section~\ref{weight_distr_section}.

\begin{definition}
    Let $\mC \subseteq \F_q^n$ be a code. The \textbf{homogeneous weight enumerator} of $\mC$ is the polynomial
    \[
    W_{\mC}(X,Y) = \sum_{i=0}^n W_i(\mC) X^{n-i} Y^i, \, \textnormal{ where } \, W_i(\mC) = |\{ c \in \mC \, : \, \omega^H(c) = i \}|.
    \]
This polynomial can also be expressed in the one-variable form $W_{\mC}(Z)$, called the \textbf{weight enumerator}, which is connected to $W_{\mC}(X,Y)$ in the following ways:
    \[
    W_{\mC}(X,Y) = X^nW_{\mC}(X^{-1}Y), \quad W_{\mC}(Z) = W_{\mC}(1,Z).
    \]
The numbers $\{W_i(\mC) \, :\, 0 \le i \le n\}$ are called the \textbf{weight distribution} of $\mC$.
\end{definition}

\begin{theorem} \label{exp_ternary_Golay}
    Let $\mC$ be the ternary Golay code of length $n=11$, dimension $k=6$, and minimum distance $d = 5$. We have
    \[
    \E[\mC] = n(H_n - H_{d-1}) - \frac{\binom{n}{k} - \frac{W_k(\mC^\perp)}{2}}{\binom{n-1}{k}} \approx 8.416.
    \] 
\end{theorem}
\begin{proof}
We have $\alpha(\mC,s) = \binom{n}{s}$ when $s \ge n-d+1 = 7$ from Corollary \ref{E[C]_general_enhanced}. Thus, we only need to compute $ \alpha(\mC,k)$. Notice that
\[
\begin{split}
    \alpha(\mC,k) &= |\{ S \subseteq \{1, \dots, n\} \, : \, S \textnormal{ is a minimal information set of $\mC$}  \}| \\
    &= |\{ S \subseteq \{1, \dots, n\} \, : \, |S| = k \, , \, \nexists \, x \in \mC^\perp \setminus \{0\} \, \textnormal{ with } \, \sigma(x) \subseteq S \}| \\
    &= \binom{n}{k} - |\{ S \subseteq \{1, \dots, n\} \, : \, |S| = k \, , \, \exists \, x \in \mC^\perp \setminus \{0\} \, \textnormal{ with } \, \sigma(x) \subseteq S \}|.
\end{split}
\]
The homogeneous weight enumerator of the ternary Golay code is 
\begin{equation} \label{ternary_golay_macwilliams}
    W_{\mC}(X,Y) = X^{11} + 132 X^6 Y^5 + 132 X^5 Y^6  + 330 X^3Y^8 + 110 X^2 Y^9 + 24 Y^{11};
\end{equation}
see~\cite[Example~7.6.2]{fundamental_err_corr_codes}. We now combine~\eqref{ternary_golay_macwilliams} with the  MacWilliams Identities: 
\[
    W_{\mC^\perp}(X,Y) = \frac{1}{|\mC|} W_{\mC}(X+(q-1)Y,X-Y);
\]
see e.g.~\cite[Chapter~5, Theorem~13]{mds2}.
We obtain that the homogeneous weight enumerator of the dual code~$\mC^{\perp}$ is
\[
W_{\mC^\perp}(X,Y) = X^{11} + 132 X^5 Y^6 + 110 X^2 Y^9.
\] 
Thus, 
\[
\begin{split}
    \alpha(\mC,k) &= \binom{n}{k} - |\{ S \subseteq \{1, \dots, n\} \, : \, |S| = k \, , \, \exists \, x \in \mC^\perp \setminus \{0\} \, \textnormal{ with } \, \sigma(x) \subseteq S \}| \\
    &= \binom{n}{k} - \frac{W_k(\mC^\perp)}{2},
\end{split}
\]
where the latter equality holds because if the support of a nonzero codeword $x \in \mC^\perp$ is contained in a $k$-subset $S \subseteq \{1, \dots, n\}$, then $2x$ is the only other nonzero codeword of $\mC^\perp$ whose support is contained in $S$. This follows from the fact that the minimum distance of the dual code $\mC^\perp$ is $d(\mC^\perp) = 6$. 
\end{proof}

The analogous result for the extended ternary Golay code can be obtained with a similar reasoning
and the proof is therefore omitted.

\begin{theorem} \label{exp_extended_ternary_Golay}
    Let $\mC$ be the extended ternary Golay code of length $n=12$, dimension $k=6$, and minimum distance $d = 6$. Then
    \[
    \E[\mC] = n(H_n - H_{d-1}) - \frac{\binom{n}{k} - \frac{W_k(\mC)}{2}}{\binom{n-1}{k}} \approx 8.124.
    \]
\end{theorem}

\section{Expectation and Weight Distributions} \label{weight_distr_section}

In the previous section we computed the value of $\E[\mC]$ for the ternary Golay code and the extended ternary Golay code using knowledge about their weight enumerator. It is natural to ask whether this is a general phenomenon and if the expectation $\E[\mC]$ can be always expressed in terms of the weight enumerator of $\mC$ (or equivalently of $\mC^\perp$). The following example provides a negative answer. 

\begin{example} \label{example_wd_different_E} 
    Let $\mC_1$ and $\mC_2$ be the binary codes of length $n=12$ and dimension $k=3$ generated by the matrices
    \[
    G_1 = \begin{pmatrix}
    1 & 1 & 1 & 0 & 0 & 0 & 0 & 0 & 0 & 0 & 0 & 0 \\
    0 & 0 & 0 & 1 & 1 & 1 & 1 & 0 & 0 & 0 & 0 & 0 \\
    0 & 0 & 0 & 0 & 0 & 0 & 0 & 1 & 1 & 1 & 1 & 1
    \end{pmatrix}
    \]
    and
    \[
    \quad G_2 = \begin{pmatrix}
    1 & 0 & 0 & 0 & 0 & 0 & 1 & 1 & 1 & 1 & 1 & 1 \\
    0 & 1 & 1 & 0 & 0 & 0 & 1 & 1 & 1 & 1 & 1 & 1 \\
    0 & 0 & 0 & 1 & 1 & 1 & 1 & 1 & 1 & 1 & 1 & 1
    \end{pmatrix},
    \]
    respectively. One can check that $\mC_1$ and $\mC_2$ have the same weight enumerator, namely,
    \[
    W_{\mC_1}(Z) = 1 + Z^3 + Z^4 + Z^5 + Z^7 + Z^8 + Z^9 + Z^{12} = W_{\mC_2}(Z).
    \]
    On the other hand, we have 
    \[
        \mathbb{E}[\mC_1] = \frac{1229}{210} \approx 5.852 \quad \textnormal{and} \quad \mathbb{E}[\mC_2] = \frac{2633}{462} \approx 5.699.
    \]
    Observe that $\mC_1$ and $\mC_2$ are inequivalent codes, as otherwise one would necessarily have 
    $\E[\mC_1]=\E[\mC_2]$.
\end{example}

At this point, one may wonder if the value of $\E[\mC]$ can be expressed in terms of some combinatorial invariant of the code $\mC$. While we have provided evidence that the weight distribution is not such an invariant, in this section we prove that a finer invariant suffices, namely, the weight distributions of the extension codes $\mC \otimes_{\F_q} \F_{q^m}$ for $1 \le m \le n$.

\begin{definition}
    Let $\mC \subseteq \F_q^n$ be a code 
    and let $m \ge 1$ be an integer. The \textbf{extension code} of $\mC$ over $\F_{q^m}$, denoted by $\mC \otimes_{\F_q} \F_{q^m}$, is the set of $\F_{q^m}$-linear combinations of the codewords of $\mC$.
\end{definition}

The rest of this section is devoted to proving the following main result connecting $\E[\mC]$ to the weight distributions of the extension codes.

\begin{theorem} \label{expectation_weight_distr}
    Let $\mC \subseteq \F_q^n$ be a code of dimension $k$ and minimum distance $d$. We have  
   \[
    \E[\mC] = n(H_n - H_{d-1}) - \sum_{r=k}^{n-d} \frac{1}{\binom{n-1}{r}} \sum_{\ell=0}^{n-r} \binom{n-\ell}{r} \sum_{m=0}^n (-1)^m W_{\ell}(\mC  \otimes_{\F_q} \F_{q^m}) \gamma(q,m,n),
\]
where 
\[
\gamma(q,m,n) = \sum_{j=m}^n \Bigg( \prod_{\nu=0}^{j-1} \frac{1}{q^j - q^\nu} \Bigg) q^{\binom{j}{2} + \binom{j-m}{2}} \qbinom{j}{m}_q, 
\]
and with the convention that $\mC  \otimes_{\F_q} \F_{q^m}$ is the zero code when $m=0$.
\end{theorem}

We find the previous result interesting in its own right. However, in Section~\ref{reedmuller_section} we will also show how it can be used to compute $\E[\mC]$, where $\mC$ is any first-order $q$-ary Reed-Muller code.

The proof of Theorem~\ref{expectation_weight_distr} relies on four lemmas that we establish separately.
Throughout the remainder of this section, for any $m \ge 1$ we fix a basis $\Omega_m=\{\omega^1_m, \ldots, \omega^m_m\}$ of $\F_{q^m}$ over~$\F_q$ and denote by $\varphi_m:\F_{q^m} \to \F_q^m$ the expansion map over the basis $\Omega_m$, which is an $\F_q$-isomorphism.

For an $\F_q$-subspace $V \subseteq \F_q^n$, we let
$\F_q^{m \times n}[V]$ denote the $\F_q$-space of $m \times n$ matrices whose rows are vectors from $V$. When $S \subseteq \{1, \ldots, n\}$, we let
$\F_q^{m \times n}[S]$ be the set of matrices whose columns indexed by any $i \notin S$ are identically zero.

The map $\varphi_m$ induces a map $\widehat{\varphi}_m: \mC \otimes_{\F_q} \F_{q^m} \to \F_q^{m \times n}[\mC]$ as follows: For $x \in \mC \otimes_{\F_q} \F_{q^m}$, write $x=(x_1, \ldots, x_n)$ into standard coordinates, and construct the matrix $M_m(x)$ whose columns are the vectors \smash{$\varphi_m(x_1)^\top, \ldots, \varphi_m(x_n)^\top$}, where the superscript ``$\top$'' denotes transposition. It can be checked that $M_m(x) \in \F_q^{m \times n}[\mC]$. We define $\widehat{\varphi}_m : x \mapsto M_m(x)$.

\begin{lemma}[\text{see \cite[Proposition~6]{ext_weight_jurrius}}] \label{isomorphism_extension_codes}
Let $x \in \mC \otimes_{\F_q} \F_{q^m}$ and $S \subseteq \{1, \ldots, n\}$. The following are equivalent:
\begin{itemize}
    \item[1.] $\sigma(x) \subseteq S$,
    \item[2.] $\widehat{\varphi}_m(x) \in \F_q^{m \times n}[S]$.
\end{itemize}
\end{lemma}

The next result recalls a well-known combinatorial identity.  It can be obtained by enumerating the $j \times r$ matrices with entries in $\F_q$ by their rank.
 
\begin{lemma}[\text{see \cite[Lemma~26]{power_of_q_inversion}}] \label{q^jm_inversion} 
    For all $j,r \in \N$ we have
    \[
    q^{jr} = \sum_{i=0}^r \qbinom{j}{i}_q \qbinom{r}{i}_q \prod_{s=0}^{i-1} (q^i - q^s).
    \]
\end{lemma}

The last two lemmas establish \textit{ad-hoc} inversion formulas that we will need in the proof of the main result.

\begin{lemma} \label{first_inversion}
Let $\{x_i\}_{i=0}^\infty$ be an integer sequence.
    Let
    \[
    y_m = \sum_{i=0}^m \qbinom{m}{i}_q x_i \quad \textnormal{for all $m \ge 0$}.
    \]
    Then
    \[
    x_i = \sum_{m=0}^i (-1)^{i-m} q^{\binom{i-m}{2}} \qbinom{i}{m}_q y_m \quad \textnormal{ for all $i \ge 0$}.
    \]
\end{lemma}
\begin{proof}
    Fix $i \ge 0$. We have
    \allowdisplaybreaks
    \[
    \begin{split}
    \sum_{m=0}^i (-1)^{i-m} q^{\binom{i-m}{2}} \qbinom{i}{m}_q y_m &= \sum_{m=0}^i (-1)^{i-m} q^{\binom{i-m}{2}} \qbinom{i}{m}_q \sum_{j=0}^m \qbinom{m}{j}_q x_j \\
    &= \sum_{j=0}^i x_j \sum_{m=j}^i (-1)^{i-m} q^{\binom{i-m}{2}} \qbinom{i}{m}_q  \qbinom{m}{j}_q \\
    &= \sum_{j=0}^i x_j \sum_{m=j}^i (-1)^{i-m} q^{\binom{i-m}{2}} \qbinom{i}{j}_q \qbinom{i-j}{i-m}_q \\
    &= \sum_{j=0}^i \qbinom{i}{j}_q x_j \sum_{m=j}^i (-1)^{i-m} q^{\binom{i-m}{2}}  \qbinom{i-j}{i-m}_q \\
    &= \sum_{j=0}^i \qbinom{i}{j}_q x_j \sum_{m=0}^{i-j} (-1)^m q^{\binom{m}{2}}  \qbinom{i-j}{m}_q =  x_i.
    \end{split}
    \]
    The latter identity holds because
    \[
    \sum_{m=0}^{i-j} (-1)^m q^{\binom{m}{2}}  \qbinom{i-j}{m}_q = 0
    \]
    unless $i = j$ (in which case the value is 1), by the $q$-Binomial Theorem.
\end{proof}

The proof of the next inversion formula is analogous to the previous one and it is therefore omitted.

\begin{lemma} \label{second_inversion}
    Let $n \ge 1$  and let 
    $x_0, x_1, \ldots, x_n$ be integers.
    Let
    \[
    y_i = \sum_{j=i}^n \qbinom{j}{i}_q x_j \quad \textnormal{for $0 \le i \le n$}.
    \]
    Then
    \[
    x_j = \sum_{i=j}^n (-1)^{i-j} q^{\binom{i-j}{2}} \qbinom{i}{j}_q y_i \quad \textnormal{for all $0 \le j \le n$.}
    \]
\end{lemma}

\begin{proof}[Proof of Theorem \ref{expectation_weight_distr}]
Throughout the proof, for $0 \le j \le k$ and $0 \le r \le n$, we let
\[
\widehat{\beta}_j(\mC,r) = |\{ S \subseteq \{1, \dots, n\} \, : \, |S| = r, \, \dim(\mC(S)) = j \}| .
\]
We double count the cardinality of the set 
\[
\mathcal{A}_m=\{ (M,S) \, : \, M \in \F_q^{m \times n} [\mC] , \, S \subseteq \{1, \dots, n\} ,\,  |S| = r, \, \sigma(\operatorname{rowsp}(M)) \subseteq S \}
\] 
for $m \ge 0$, where $\F_q^{0 \times n}=\{\emptyset\}$ and $\operatorname{rowsp}(\emptyset)=\{0\}$ by convention. On the one hand,
\[
\begin{split}
    |\mathcal{A}_m| &= \sum_{\substack{S \subseteq \{1, \dots, n\} \\ |S| = r}} |\{ M \in \F_q^{m \times n} [\mC] \, : \, \sigma(\operatorname{rowsp}(M)) \subseteq S \}| \\
    &= \sum_{\substack{S \subseteq \{1, \dots, n\} \\ |S| = r}} |\mC(S)|^m = \sum_{j=0}^n \sum_{\substack{S \subseteq \{1, \dots, n\} \\ |S| = r \\ \dim(\mC(S)) = j}} q^{jm} = \sum_{j=0}^n q^{jm} \widehat{\beta}_j(\mC,r).
\end{split}
\]
On the other hand,
\[
\begin{split}
    |\mathcal{A}_m|   &= \sum_{M \in \F_q^{m \times n} [\mC]} |\{ S \subseteq \{1, \dots, n\}\, : \,   |S| = r, \, \sigma(\operatorname{rowsp}(M)) \subseteq S \}| \\
    &= \sum_{x \in \mC \otimes_{\F_q} \F_{q^m}} |\{ S \subseteq \{1, \dots, n\}\, : \,   |S| = r, \, \sigma(x) \subseteq S \}| \\
    &= \sum_{\ell=0}^n \sum_{\substack{x \in \mC \otimes_{\F_q} \F_{q^m} \\ \omega^H(x) = \ell}} \binom{n-\ell}{r - \ell} = \sum_{\ell=0}^n \binom{n-\ell}{r - \ell}  W_{\ell}(\mC \otimes_{\F_q}  \F_{q^m}),
\end{split}
\]
where the second-to-last equality follows from Lemma \ref{isomorphism_extension_codes}.
Therefore
\begin{equation} \label{first_equality_weight_distr}
    \sum_{j=0}^n q^{jm} \widehat{\beta}_j(\mC,r) = \sum_{\ell=0}^n \binom{n-\ell}{r - \ell}  W_{\ell}(\mC \otimes_{\F_q} \F_{q^m})
\end{equation}
for all $m \ge 0$. 
Combining \eqref{first_equality_weight_distr} with Lemma \ref{q^jm_inversion}  gives
\[
    \sum_{i=0}^m \qbinom{m}{i}_q \Bigg(\prod_{s=0}^{i-1} (q^i - q^s) \Bigg) \sum_{j=0}^n  \qbinom{j}{i}_q    \widehat{\beta}_j(\mC,r) = \sum_{\ell=0}^n \binom{n-\ell}{r - \ell}  W_{\ell}(\mC \otimes_{\F_q} \F_{q^m}).
\]
Notice that the previous identity is of the form
\[
\sum_{i=0}^m \qbinom{m}{i}_q x_i = y_m,
\]
which allows us to apply Lemma \ref{first_inversion} and obtain
\[
\sum_{j=m}^n  \qbinom{j}{m}_q    \widehat{\beta}_j(\mC,r) = \Bigg(\prod_{s=0}^{m-1} \frac{1}{q^m - q^s}\Bigg)\sum_{i=0}^m (-1)^{m-i} q^{\binom{m-i}{2}} \qbinom{m}{i}_q \sum_{\ell=0}^n \binom{n-\ell}{r - \ell}  W_{\ell}(\mC \otimes_{\F_q} \F_{q^i})
\]
for all $m \ge 0$. Therefore we can apply Lemma \ref{second_inversion} and get
\[
\begin{split}
    \widehat{\beta}_m&(\mC,r) = \\ &\sum_{j=m}^n (-1)^{j-m} q^{\binom{j-m}{2}} \qbinom{j}{m}_q \Bigg( \prod_{s=0}^{j-1} \frac{1}{q^j - q^s} \Bigg) \sum_{i=0}^j (-1)^{j-i} q^{\binom{j-i}{2}} \qbinom{j}{i}_q \sum_{\ell=0}^n \binom{n-\ell}{r - \ell}  W_{\ell}(\mC \otimes_{\F_q} \F_{q^i})
\end{split}
\]
for all $m \ge 0$.

To conclude, since $ \widehat{\beta}_0(\mC,r) = \alpha(\mC,n-r)$, we get 
\[
\begin{split}
    \alpha(\mC,n-r) &= \sum_{j=0}^n (-1)^{j} q^{\binom{j}{2}} \Bigg( \prod_{s=0}^{j-1} \frac{1}{q^j - q^s} \Bigg) \sum_{i=0}^j (-1)^{j-i} q^{\binom{j-i}{2}} \qbinom{j}{i}_q \sum_{\ell=0}^n \binom{n-\ell}{r - \ell}  W_{\ell}(\mC \otimes_{\F_q} \F_{q^i}) \\
    &= \sum_{\ell=0}^n \binom{n-\ell}{r - \ell} \sum_{i=0}^n (-1)^i W_{\ell}(\mC \otimes_{\F_q} \F_{q^i}) \sum_{j=i}^n \Bigg( \prod_{s=0}^{j-1} \frac{1}{q^j - q^s} \Bigg)  q^{\binom{j}{2} + \binom{j-i}{2}} \qbinom{j}{i}_q,
\end{split}
\]
or, equivalently,
\[
    \alpha(\mC,r)= \sum_{\ell=0}^{n-r}  \binom{n-\ell}{r}  \sum_{i=0}^n (-1)^i W_{\ell}(\mC \otimes_{\F_q} \F_{q^i}) \sum_{j=i}^n \Bigg( \prod_{s=0}^{j-1} \frac{1}{q^j - q^s} \Bigg) q^{\binom{j}{2} + \binom{j-i}{2}} \qbinom{j}{i}_q.
\]
Combining the latter with the general formula in Corollary \ref{E[C]_general_enhanced} we get the desired result.
\end{proof}

    Note that Theorem \ref{expectation_weight_distr} implies that, for any code $\mC \subseteq \F_q^n$, knowing all weight distributions of the extension codes $\mC \otimes_{\F_q} \F_{q^i}$, for $0 \le i \le n$, is sufficient to determine the expectation~$\E[\mC]$.

\section{Performance of First-Order Reed-Muller Codes} \label{reedmuller_section}
In this section we show a direct application of Theorem~\ref{expectation_weight_distr}, establishing a closed formula for 
$\E[\mC]$ where $\mC$ is any 
first-order $q$-ary Reed-Muller code.
Recall that, for an integer $s \ge 2$, the first-order $q$-ary Reed-Muller code $\mC \subseteq \F_q^n$ has length $n = q^{s-1}$, dimension $k=s$, and minimum distance $d = (q-1)q^{s-2}$; see~\cite[Definition~15]{jurrius2}.
We start by introducing the extended weight enumerator, which we will need in the proof of the main result.

\begin{definition}
    Let $\mC \subseteq \F_q^n$ be a code and $J \subseteq \{1, \dots, n\}$. Let
    \[
    B_J(U) = U^{\dim(\mC(J^{\mathsf{c}}))}-1, \quad  \quad B_t(U) = \sum_{\substack{J \subseteq \{1, \dots, n\} \\ |J| =t}} B_J(U).
    \]
    The \emph{extended weight enumerator} of $\mC$ is
    \[
    W_{\mC}(X,Y,U) = X^n + \sum_{t=0}^n B_t(U) (X-Y)^t Y^{n-t}.
    \]
\end{definition}

It turns out that
the extended weight enumerator
of first-order $q$-ary Reed-Muller codes 
is known and has been explicitly computed as follows.

\begin{lemma}[\text{see \cite[Theorem~18]{jurrius2}}]  
Let $\mC \subseteq \F_q^n$ be the first-order $q$-ary Reed-Muller code of dimension $s$, where $n = q^{s-1}$.
    Its extended weight enumerator
    is equal to
    \begin{equation} \label{ext_enum_RM}
        \begin{split}
        W_{\mC}(X,Y,U) = &\sum_{t=1}^s  \Bigg( \prod_{j=0}^{t-1} (U - q^j) \Bigg) \qbinom{s-1}{t-1}_q Y^{q^{s-1}} \\
        &+ \sum_{t=0}^{s-1} \Bigg( \prod_{j=0}^{t-1} (U - q^j) \Bigg) q^t \qbinom{s-1}{t}_q X^{q^{s-1-t}} Y^{q^{s-1}-q^{s-1-t}}. 
    \end{split}
    \end{equation}
\end{lemma}

We now have all the ingredients to solve the DNA coverage depth problem for first-order $q$-ary Reed-Muller codes.
\begin{theorem} \label{exp_RM}
    Let $\mC \subseteq \F_q^n$ be the first-order $q$-ary Reed-Muller code of dimension $s$,
    where $n = q^{s-1}$ and $d=(q-1)q^{s-2}$. We have
    \[
    \begin{split}
    \E[\mC] &= n(H_n - H_{d-1}) - \sum_{r=s}^{q^{s-2}} \frac{1}{\binom{q^{s-1}-1}{r}} \sum_{t=0}^{\lfloor s-1-\log_q(r) \rfloor} \binom{q^{s-1-t}}{r} \\
    & \qquad \cdot \sum_{i=0}^{q^{s-1}} (-1)^i  \Bigg( \prod_{j=0}^{t-1} (q^i - q^j) \Bigg) q^t \qbinom{s-1}{t}_q \gamma(q,i,q^{s-1}),
    \end{split}
    \]
where 
\[
\gamma(q,i,q^{s-1}) = \sum_{j=i}^{q^{s-1}} \Bigg( \prod_{\nu=0}^{j-1} \frac{1}{q^j - q^\nu} \Bigg) q^{\binom{j}{2} + \binom{j-i}{2}} \qbinom{j}{i}_q.
\]
\end{theorem}

\begin{proof}
We start by 
connecting the 
extended weight enumerator 
with extension codes.
More precisely, we have that
the extended weight enumerator of a code $\mC$ for $U = q^m$ coincides with the weight enumerator of the extension code $\mC \otimes_{\F_q} \F_{q^m}$:
   \begin{equation} \label{ext_weight_enum_relation}
        W_{\mC}(X,Y,q^m) = W_{\mC \otimes_{\F_q} \F_{q^m}}(X,Y) \quad \mbox{for all $m \ge 0$;}
    \end{equation}
see~\cite[Proposition 5]{ext_weight_jurrius}.
Therefore, 
combining \eqref{ext_weight_enum_relation} with \eqref{ext_enum_RM} we obtain the weight enumerator of the extension code $\mC \otimes_{\F_q} \F_{q^m}$ for a first-order $q$-ary Reed-Muller code $\mC$ of dimension $s$:
\[
\begin{split}
W_{\mC \otimes_{\F_q} \F_{q^m}}(X,Y) = &\sum_{t=1}^s  \Bigg( \prod_{j=0}^{t-1} (q^m - q^j) \Bigg) \qbinom{s-1}{t-1}_q Y^{q^{s-1}} \\
        &+ \sum_{t=0}^{s-1} \Bigg( \prod_{j=0}^{t-1} (q^m - q^j) \Bigg) q^t \qbinom{s-1}{t}_q X^{q^{s-1-t}} Y^{q^{s-1}-q^{s-1-t}}
\end{split}
\]
for all $m \ge 0$.
For the purposes of this proof,
we will need the coefficients 
\begin{equation} \label{extension_RM_coeff}
    W_{q^{s-1}-q^{s-1-t}}(\mC \otimes_{\F_q} \F_{q^m}) = \Bigg( \prod_{j=0}^{t-1} (q^m - q^j) \Bigg) q^t \qbinom{s-1}{t}_q, \quad \textnormal{for $0 \le t \le s-1.$}  
\end{equation}
We know from Theorem \ref{expectation_weight_distr} that
\begin{equation} \label{final_eq_RM}
    \sum_{r=k}^{n-d} \frac{1}{\binom{n-1}{r}} \sum_{\ell=0}^{n-r} \binom{n-\ell}{r} \sum_{i=0}^n (-1)^i W_{\ell}(\mC  \otimes_{\F_q}  \F_{q^i}) \gamma(q,i,n) =  n(H_n - H_{d-1}) -\E[\mC].
\end{equation}
By~\eqref{extension_RM_coeff}, the left-hand side of~\eqref{final_eq_RM} is  
\[
\begin{split}
    \sum_{r=k}^{n-d}& \frac{1}{\binom{n-1}{r}} \sum_{\ell=0}^{n-r} \binom{n-\ell}{r} \sum_{i=0}^n (-1)^i W_{\ell}(\mC \otimes_{\F_q} \F_{q^i}) \gamma(q,i,n) \\
    &= \sum_{r=s}^{q^{s-2}} \frac{1}{\binom{q^{s-1}-1}{r}}\sum_{\substack{0\le \ell \le q^{s-1}-r, \\ \ell = q^{s-1} - q^{s-1-t}, \\ 0 \le t \le s-1}} \binom{q^{s-1}-\ell}{r} \sum_{i=0}^{q^{s-1}} (-1)^i  \Bigg( \prod_{j=0}^{t-1} (q^i - q^j) \Bigg) q^t \qbinom{s-1}{t}_q \gamma(q,i,q^{s-1}),
\end{split}
\]
which can be re-written as
\[
\sum_{r=s}^{q^{s-2}} \frac{1}{\binom{q^{s-1}-1}{r}}\sum_{t=0}^{\lfloor s-1-\log_q(r) \rfloor} \binom{q^{s-1-t}}{r} \sum_{i=0}^{q^{s-1}} (-1)^i  \Bigg( \prod_{j=0}^{t-1} (q^i - q^j) \Bigg) q^t \qbinom{s-1}{t}_q \gamma(q,i,q^{s-1})
\]
because for any $ 0 \le t \le s-1$, the condition $0 \le q^{s-1} - q^{s-1-t} \le q^{s-1}-r$ is equivalent 
to $0 \le t \le \lfloor s-1-\log_q(r) \rfloor$. 
Combining this with~\eqref{final_eq_RM} gives the desired result.
\end{proof}

\section{Conclusions and Future Work}
In this paper, we focused on the coverage depth problem in DNA data storage from a coding theory perspective. We derived closed-form expressions for the expectation of various code families by applying techniques ranging from combinatorics to duality theory, and by exploring connections with weight distributions and extension codes. 
Beyond the results presented in this work, several promising research directions remain open. Two of them are the following:

\begin{itemize}
\item[1.] We conjectured that simplex codes solve Problem~\ref{second_problem} for the parameters for which they exist. An interesting challenge is to characterize the codes that offer the best performance in regimes where neither MDS nor simplex codes exist.
\item[2.] Although we presented various closed-form expressions for the expectation, deriving such formulas is generally difficult. Developing general lower bounds or approximation techniques would therefore be valuable. 
\end{itemize}

\end{document}